\documentclass[a4paper]{llncs}
\usepackage{color}
\usepackage[linesnumbered,boxed]{algorithm2e}
\usepackage{graphicx}
\usepackage{latexsym}
\usepackage{amsmath}
\usepackage{amssymb}
\usepackage{enumerate}

\begin{document}

\title{Algorithmic and Hardness Results for the Colorful Components Problems}
\author{Anna Adamaszek\inst{1} \and Alexandru Popa\inst{2}}

\institute{Max-Planck-Institut f\"ur Informatik, Saarbr\"ucken, Germany,\\        \texttt{anna@mpi-inf.mpg.de} 
\and Faculty of Informatics, Masaryk University, Brno, Czech Republic,\\ \texttt{popa@fi.muni.cz}
}

\maketitle
\date{ }

\begin{abstract}
In this paper we investigate the \emph{colorful components} framework, motivated by applications emerging from comparative genomics~\cite{Sankoff11}. The general goal is to remove a collection of edges from an undirected vertex-colored graph $G$ such that in the resulting graph $G'$ all the connected components are \emph{colorful} (i.e., any two vertices of the same color belong to different connected components). We want $G'$ to optimize an objective function, the selection of this function being specific to each problem in the framework.

We analyze three objective functions, and thus, three different problems, which are believed to be relevant for the biological applications: minimizing the number of singleton vertices, maximizing the number of edges in the transitive closure, and minimizing the number of connected components.

Our main result is a polynomial time algorithm for the first problem. This result disproves the conjecture of Zheng et al.~\cite{ZSLS11} that the problem is $ NP$-hard (assuming $P \neq  NP$). Then, we show that the second problem is $ APX$-hard, thus proving and strengthening the conjecture of Zheng et al.~\cite{ZSLS11} that the problem is $ NP$-hard. Finally, we show that the third problem does not admit polynomial time approximation within a factor of $|V|^{1/14 - \epsilon}$ for any $\epsilon > 0$, assuming $P \neq NP$ (or within a factor of $|V|^{1/2 - \epsilon}$, assuming $ZPP \neq NP$).
\end{abstract}

\section{Introduction}

In this paper we consider the following framework.

\vskip10pt
{\sc Colorful components framework:}
Given a simple, undirected graph $G=(V,E)$, and a coloring $c:V \to C$ of the vertices with colors from a given set $C$, remove a collection of edges $E' \subseteq E$ from the graph such that each connected component in $G' = (V, E \backslash E')$ is a \emph{colorful component} (i.e., it does not contain two identically colored vertices). We want the resulting graph $G'$ to be optimal according some fixed \emph{optimization measure}.
\vskip10pt

In this paper we consider three optimization measures and, respectively, three different problems: \emph{Minimum Singleton Vertices (MSV)}, \emph{Minimum Edges in Transitive Closure (MEC)}, and \emph{Minimum Colorful Components (MCC)}.
We now introduce the optimization measures for all these problems. 

\begin{problem}[Minimum Singleton Vertices]
The goal is to minimize the number of connected components of $G'$ that consist of one vertex.
\end{problem}


\begin{problem}[Maximize Edges in Transitive Closure]
The goal is to maximize the number of edges in the transitive closure of $G'$. 
\end{problem}

If a graph consists of $k$ connected components, each containing respectively $a_1, a_2, \dots, a_k$ vertices, the number of edges in the transitive closure equals 
$$ \sum_{i=1}^k \frac{a_i \cdot (a_i - 1)}{2} \enspace.$$


\begin{problem}[Minimum Colorful Components]
The goal is to minimize the number of connected components in $G'$.
\end{problem}

The first two problems have been introduced in~\cite{ZSLS11}, while the third one is newly introduced in this paper.

\paragraph*{Motivation.}

The colorful components framework is motivated by applications originating from comparative genomics~\cite{Sankoff11,ZSLS11}, which is a fundamental branch of bioinformatics that studies the relationship of the genome structure between different biological species. The information achieved from this field can help scientists to 
improve the understanding of the structure and the functions of human genes and, consequently, find treatments for many diseases~\cite{Mushegian10}.  

As pointed out in~\cite{Sankoff11,ZSLS11}, one of the key problems in this area, the multiple alignment of gene orders, can be captured as a graph theoretical problem, using the colorful components framework. We refer the reader to~\cite{ZSLS11}  for an overview of the connection between the multiple alignment of gene orders and the graph theoretic framework considered, and for a discussion about the biological motivation of two particular problems we consider, Minimum Singleton Vertices and Maximize Edges in Transitive Closure.

\paragraph*{Related work.}

We now discuss the collection of known problems which fit into the connected components framework. 

We start with a problem named either \emph{Colorful Components}~\cite{BHKNTU12,BHKN13} or \emph{Minimum Orthogonal Partition}~\cite{HLZ00,ZSLS11}, since this problem has received the most attention so far. In this problem the objective function is to minimize the number of edges removed from $G$ to obtain the graph $G'$ in which all the components are colorful. Bruckner et al. show ~\cite{BHKNTU12} that the problem is $ NP$-hard for three or more colors and they study fixed parameter tractable algorithms for the problem. Their $ NP$-hardness reduction can be modified slightly (starting the reduction from a version of 3SAT when each variable occurs only $O(1)$ times, instead of from the general 3SAT) to show the APX-hardness of the problem.  Zheng et al.~\cite{ZSLS11} and Bruckner et al.~\cite{BHKN13} study heuristic approaches for the problem, and He et al.~\cite{HLZ00} present an approximation algorithm for some special case of the problem. As the general problem is a special case of the Minimum Multi-Multiway Cut, it admits a $O(\log |C|)$ approximation algorithm~\cite{AL07}.


Other objective functions have been proposed, with the hope that some of them are both tractable and biologically meaningful. The MSV and the MEC problems have been introduced by Zheng et al.~\cite{ZSLS11}, who presented heuristic algorithms for the problems, without giving any worst-case approximation guarantee. They also conjectured both problems to be NP-hard. We are not aware of any other results concerning the MSV and MEC problems, or of any previous research on the MCC problem.

%

\paragraph*{Our results.}

Our main result is a polynomial time \emph{exact} algorithm for the MSV problem, presented in Section~\ref{sec:poly_algo}. This disproves the conjecture of Zheng et al.~\cite{ZSLS11} that the problem is $ NP$-hard (assuming $ P \neq  NP$).
Our algorithm maintains a feasible solution $G'=(V,E')$ for the MSV problem, starting with an empty graph $G'=(V,\emptyset)$. Then, in each step $G'$ is modified by applying to it a carefully chosen alternating path $p$, starting at a singleton vertex. The alternating path consists of the edges of $G$, and its every second edge is in $G'$. Applying $p$ to $G'$ means that the edges from $p$ which are not in $G'$ are added to $G'$, and at the same time the edges of $p$ which are in $G'$ are removed from $G'$. The algorithm ensures that at each step $G'$ is a feasible solution to the problem, and satisfies an invariant that all connected components in $G'$ are either singletons, edges or stars. In the analysis we show that when the algorithm does not find any new alternating path, the number of singleton components in $G'$ matches the lower bound presented in Section \ref{subsec:msv_lb}.

In Section~\ref{MEC} we study the MEC problem and we show that the problem is $ NP$-hard and $ APX$-hard when the number of colors in the graph is at least $4$. This proves the conjecture of Zheng et al~\cite{ZSLS11}. We show the result via a reduction from the version of the MAX-3SAT problem, where each variable appears at most some constant number of times in the formula (see~\cite{Ausiello1999}, Section 8.4).

Finally, in Section~\ref{MCC} we consider the MCC problem, which is introduced for the first time in this paper. We prove that MCC does not admit polynomial time approximation within a factor of $|V|^{1/14 - \epsilon}$, for any $\epsilon > 0$, unless $P = NP$ (or within a factor of $|V|^{1/2 - \epsilon}$, unless $ZPP = NP$), even if each vertex color appears at most two times. We show the inapproximability result via a reduction from the Minimum Clique Partition problem which is equivalent to Minimum Graph Coloring~\cite{PM81}.

Due to space constraints some proofs have been moved to the appendix.

\section{A Polynomial Time Exact Algorithm for the MSV}
\label{sec:poly_algo}

In this section we present a polynomial time algorithm {\sc MSVexact} which finds an optimal solution for the MSV problem. First, in Section \ref{subsec:msv_lb} we show a lower bound on the number of singleton vertices in any feasible solution for the problem. Then, in Section \ref{subsec:msv_alg} we describe the algorithm, with its key procedure presented in Section \ref{subsec:msv_alternating_path}. The analysis of the algorithm is then performed in Section \ref{subsec:msv_analysis}.

\subsection{Lower Bound}\label{subsec:msv_lb}

Let $G=(V,E)$, together with a coloring $c: V \to C$, be an input instance for the MSV problem.
For any color $c$ let $V_c \subseteq V$ denote the set of vertices of color $c$. For any set of vertices $V' \subseteq V$ we denote by $N(V')$ the set of neighbors of $V'$ in $G$, i.e. $N(V') = \{v \in V \setminus V': \exists v' \in V'  \ (v',v) \in E \}$. For any set of colors $C' \subseteq C$ and set of vertices $V' \subseteq V$ we denote by $N_{C'}(V')$ the set of neighbors of $V'$ in $G$ which have colors in $C'$, i.e. $N_{C'}(V') = \{ v \in N(V'): c(v) \in C'\}$.

\begin{lemma}
\label{lem:MSV_lb}
For any color $c$ let 
$$s_c = \max_{V' \subseteq V_c} (|V'| -|N_{C \setminus \{c\}}(V')|)\enspace.$$
Then in every feasible solution for the MSV problem there are at least $s_c$ singletons of color $c$.
\end{lemma}

\begin{proof}
Let $G' = (V,E')$, where $E' \subseteq E$, be a feasible solution for $G$. Fix a color $c$ for which $s_c > 0$ and let $V' \subseteq V_c$ be the subset maximizing the value of $s_c$. (Notice that $s_c$ depends only on the graph $G$, and not on $G'$.) For each vertex $v' \in V'$ which is not a singleton in $G'$ we pick an arbitrary neighbor $n(v')$ in $G'$. We have $n(v') \in N_{C \setminus \{c\}}(V')$. As any two vertices from $V'$ belong to different connected components in $G'$, the vertices $n(v')$ are pairwise different. The number of vertices of $V'$ which are not singletons in $G'$ is therefore at most $|N_{C \setminus \{c\}}(V')|$. The number of singletons amongst vertices from $V'$, and therefore also the number of singletons of color $c$, is therefore at least  $|V'| -|N_{C \setminus \{c\}}(V')| = s_c$.  \qed
\end{proof}

\begin{corollary}\label{cor:MSV_lb}
In any feasible solution for the MSV problem there are at least $\sum_{c \in C}\ s_c$ singleton vertices.
\end{corollary}

\subsection{Idea of the Algorithm}\label{subsec:msv_alg}

We now present an algorithm {\sc MSVexact} which finds an optimal solution for the MSV problem. The input of the algorithm consists of a simple, undirected graph $G=(V,E)$, together with a coloring $c: V \to C$. The algorithm maintains a feasible solution $G'=(V,E')$ for the MSV problem (i.e., $G'$ is a subgraph of the input graph $G$, and every connected component of $G'$ is a colorful component), starting with an empty graph $G'=(V,\emptyset)$. In each step the graph $G'$ is modified by adding to it a carefully chosen alternating path $p$. The alternating path consists of the edges of $G$, and its every second edge is in $G'$. Applying $p$ to $G'$ means that the edges from $p$ which are not in $G'$ are added to $G'$, and at the same time the edges of $p$ which are in $G'$ are removed from $G'$.
See Algorithm \ref{alg:general} for the formal description of the algorithm.

The path $p$ is chosen in such a way, that adding it to $G'$ decreases the number of singleton vertices of color $c$, without increasing the number of singleton vertices of other colors. Additionally, at each step of the algorithm the graph $G'$ satisfies an invariant, that each connected component of $G'$ is a singleton vertex, an edge, or a star (where a star is a tree of diameter $2$, in particular it has at least $3$ vertices).

\begin{algorithm}[t]
\DontPrintSemicolon
\KwIn{A simple, undirected graph $G=(V,E)$, a coloring $c:V \to C$}
\KwOut{A subgraph of $G$ minimizing the number of connected components, and in which each connected component is colorful}
$G' := (V, \emptyset)$\;
\ForEach{$c \in C$} {
 \While{p={\sc Alternating\_Path(G',C,G,c)} \textrm{is a path}} {
 \textrm{apply $p$ to $G'$}
 }
}
\caption{{\sc MSVexact(G,c)}}\label{alg:general}
\end{algorithm}

We will show that when the algorithm stops, i.e., when it does not find any alternating path $p$ which can be added to $G'$ to decrease the number of singletons of any color, the number of singleton vertices in $G'$ matches the lower bound from Corollary~\ref{cor:MSV_lb}.

\subsection{Finding an Alternating Path}\label{subsec:msv_alternating_path}
\SetAlgorithmName{Procedure}{procedure}{List of Procedures}

Let $G'=(V,E')$ be a feasible solution for an instance $(G=(V,E), C)$ of the MSV problem, such that each connected component of $G'$ is a singleton vertex, an edge, or a star. Let $c \in C$ be an arbitrary color, and let $S_c \subseteq V$ be the set of all singletons of color $c$ in $G'$. We  describe a procedure {\sc Alternating\_Path(G,C,G',c)} which outputs an alternating path $p$ for $G'$ in $G$. In the following section we prove that the path $p$ satisfies the properties outlined in Section \ref{subsec:msv_alg}, and that when no path is found, the number of singletons of color $c$ in $G'$ matches the lower bound from Lemma~\ref{lem:MSV_lb}.

The idea behind the path construction is as follows. We want to find a path starting in some singleton vertex of color $c$, connecting each vertex of color $c$ with a vertex of color different than $c$ using an edge $e \in E \setminus E'$; and each vertex of color different than $c$ with an vertex of color $c$ using an edge $e \in E'$. We end the construction of the path when the current endpoint $v \notin V_c$ of the path belongs to a connected component of $G'$ to which we can attach an additional vertex of color $c$ (possibly while splitting the component into two parts). Such a case occurs when $v$ is a leaf of a star (which will result in removing $v$ from the star-component and connecting it with the vertex of color $c$), or when the connected component of $v$ does not contain color $c$. Then applying the alternating path to the graph $G'$ results in ``switching'' vertices of color $c$ between different connected components of $G'$, and removing one singleton of color $c$, as the start point of the path will not be a singleton in the new graph. The algorithm performs a BFS
search of the path satisfying the required conditions, starting with the collection of all singleton vertices of color $c$. See Procedure \ref{alg:alternating_path} for a formal description of the procedure.

\begin{algorithm}[t]
\DontPrintSemicolon
\KwIn{A simple, undirected graph $G=(V,E)$, a coloring $c:V \to C$, a feasible subgraph $G'=(V,E')$ of $G$, and a color $c \in C$}
\KwOut{A path $p$ or {\sc no\_path\_found}} 
$V' := S_c$\;
$N' := N_{C \setminus \{c\}}(V')$\tcp*[r]{Neighbors in $G$}
$\forall v \in N' \textrm{ pred}(v) := \textrm{ any }v' \in S_c \textrm{ s.t. }(v,v') \in E$\;
\While{$|N'| > 0$} {
  \If{$\exists v \in N' : v $ is a leaf of a star in $G'$} {
    $p:=${\sc{Path\_From(v)}}\;
    \Return $p \cup (v,v')$ s.t. $(v,v') \in E'$\;
  }
  \If{$\exists v \in N' : $ the connected component of $v$ in $G'$ has no color $c$} {
    $p:=${\sc{Path\_From(v)}}\;
    \Return $p$\;
  }
  $V'':=\{v'' \in V_c: \exists v \in N'$ s.t. $(v,v'') \in E'\}$\;
  $\forall v'' \in V'' \textrm{pred}(v'') := \textrm{any }v \in N' \textrm{ s.t. }(v,v'') \in E'$\;
  $V' := V' \cup V''$\;
  $N' := N_{C \setminus \{c\}}(V') \setminus N_{C \setminus \{c\}}(V' \setminus V'')$\;
  $\forall v \in N' \textrm{ pred}(v) := \textrm{any }v' \in V'' \textrm{ s.t. }(v,v') \in E$\;
}
\Return {\sc no\_path\_found}
\caption{{\sc Alternating\_Path(G,C,G',c)}}\label{alg:alternating_path}
\end{algorithm}

\begin{algorithm}[t]
\DontPrintSemicolon
\KwIn{A vertex $v \in V$}
\KwOut{A path starting in $S_c$ and ending in $v$}
\If {pred($v$) $\in S_c$}{
 \Return (pred(v),v)
}
\Return {\sc Path\_From}(pred(v)) $\cup$ (pred(v),v)
\caption{{\sc Path\_From(v)}}\label{alg:path_from}
\end{algorithm}


Procedure {\sc Alternating\_Path} constructs the path $p$ as follows. It keeps a set of vertices $V'$ of color $c$, initially setting $V':=S_c$ (line $1$). For each element $v \notin S_c$ considered by the procedure, its \emph{predecessor} pred($v$) is fixed (line $3,14,17$). Intuitively pred($v$) is an element such that $(\textrm{pred}(v),v) \in E$, and processing pred($v$) by the procedure resulted in adding $v$ to one of the sets $V',N'$. Procedure {\sc Path\_From(v)}, invoked in lines $6$ and $10$, can then reconstruct the whole path, starting from the final vertex $v$ and finding the predecessors until it reaches a vertex from $S_c$ (see Procedure \ref{alg:path_from} for a formal description).

Each loop of the algorithm (lines $4$ -- $18$) considers the set $N'$ of new neighbors of the vertices from $V'$ (i.e., the neighbors of $V'$ which have not been considered in the previous loops), see lines $2$ and $16$, in search for vertices which can yield an end of the path (see lines $5,9$). If no such vertex is found, the set $V'$ will be further increased to include the neighbors of $N'$ of color $c$ (line $13,15$). The process continues until an appropriate vertex $v$ is found in $N'$ (lines $5,9$), and then the algorithm returns the path reconstructed from $v$, or the set $N'$ becomes empty, in which case the answer {\sc no\_path\_found} is returned (line $19$). 

\subsection{Analysis}\label{subsec:msv_analysis}

\begin{lemma}\label{lem:msv_analysis_stop}
When the procedure {\sc Alternating\_Path}$(G,C,G',c)$ invoked for a graph $G'$ which is a feasible solution for MSV for $G=(V,E)$, and s.t. each connected component of $G'$ is a singleton, an edge or a star returns {\sc no\_path\_found}, then $|S_c| = s_c$.
\end{lemma}

\begin{lemma}\label{lem:msv_analysis_graph}
Let $G'=(V,E')$ be a feasible solution for MSV for $G=(V,E)$, s.t. each connected component of $G'$ is a singleton, an edge or a star. Let $p$ be a path returned by {\sc Alternating\_Path}$(G,C,G',c)$ for some color $c$, and let $G''$ be the result of applying $p$ on $G'$. Then:
\begin{enumerate}[a)]
\item $p$ is an alternating path for $G'$ in $G$,
\item the number of singleton vertices of color $c$ in $G''$ is smaller than in $G'$; the number of singleton vertices of any other color does not increase,
\item each connected component of $G''$ is colorful,
\item each connected component of $G''$ is a singleton, an edge or a star.
\end{enumerate}
\end{lemma}

We now have all tools to prove the main theorem of this section.

\begin{theorem}
The algorithm {\sc MSVexact(G,c)} finds an optimal solution for the MSV problem in polynomial time.
\end{theorem}
\begin{proof}
The algorithm {\sc MSVexact(G,c)} starts by choosing a feasible solution $G'=(V,\emptyset)$ for the problem, in which every connected component is a singleton. Lemma~\ref{lem:msv_analysis_graph} implies that after each step of executing the procedure {\sc Alternating\_Path}$(G,C,G',c)$, the new graph $G'$ obtained is a feasible solution (Lemma \ref{lem:msv_analysis_graph}$c$) which is a collection of singletons, edges and stars (Lemma \ref{lem:msv_analysis_graph}$d$). As in each step where finding an alternating path has been successful the number of singleton vertices of the currently processed color $c' \in C$ decreases, and for other colors does not increase (Lemma~\ref{lem:msv_analysis_graph}$b$), after $O(|V|)$ steps the algorithm does not find any more alternating paths. Thus, as each color $c' \in C$ has been processed by the algorithm, from Lemma~\ref{lem:msv_analysis_stop} for each color $c' \in C$ the number of singleton vertices of color $c'$ equals $s_{c'}$ and the resulting graph $G'$ is  an optimal solution to the MSV problem (see Corollary \ref{cor:MSV_lb}). As each execution of the procedure {\sc Alternating\_Path} takes polynomial time (as in each loop, possibly except of the last one, the set $V'$ grows), the running time of the algorithm {\sc MSVexact(G,c)} polynomial in the size of the input graph $G$. \qed
\end{proof}

\section{Hardness of MEC}
\label{MEC}
In this section we prove the $ NP$-hardness and the $ APX$-hardness of the MEC problem, for $|C| \ge 4$. We show our result via a reduction from MAX-3SAT($\beta$), a version of the MAX-3SAT problem where each variable appears at most $\beta$ times in the formula. For $\beta=3$ the problem is $APX$-hard (see~\cite{Ausiello1999}, Section 8.4).

\subsection{Reduction from MAX-3SAT($\beta$)}

Given an instance of the MAX-3SAT($\beta$) problem, i.e., a 3-CNF formula $\phi$ with $m$ clauses and $n$ variables, where each variable appears at most $\beta$ times, we  construct an instance of the MEC problem. Our instance is a vertex colored graph $G=(V,E)$, where the vertices are colored with colors from a four-element set $\{a,b,c,v\}$. An example of the reduction is illustrated in Figure~\ref{fig:mec-reduction}. 

\begin{figure}[t]
\begin{center}
\includegraphics[scale=1]{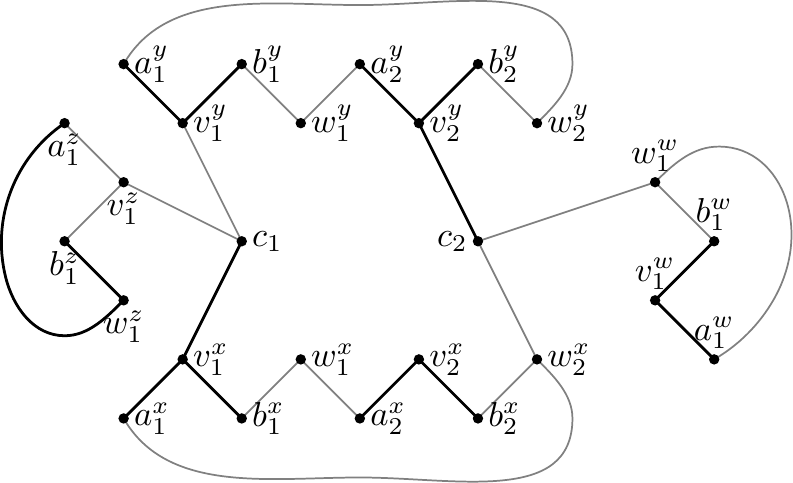}
\end{center}
\caption{An instance $G$ of the MEC problem corresponding to the 3SAT formula $(x \lor y \lor z) \land (\neg x \lor y \lor \neg w)$ (both black and gray edges). The subgraph $G''$ consisting of all vertices and only black edges represents a solution for $G$ corresponding to the following assignment: $f(x)=f(y)=f(w)=\textrm{TRUE}$, $f(z)=\textrm{FALSE}$.}
\label{fig:mec-reduction}
\end{figure}

\vskip 5pt
\noindent First we describe the set of vertices $V$.

\begin{enumerate}
\item We add to $V$ a set of vertices $c_1, \ldots, c_m$, each colored with color $c$, where vertex $c_i$ corresponds to the $i$-th clause of the formula.
\item For a variable $x$, let $n_x$ be the number of occurrences of the literals $x$ and $\neg x$ in the formula. For each variable $x$, we add to $V$: $n_x$ vertices of color $a$ (denoted by $a^x_1, a^x_2, \dots, a^x_{n_x}$), $n_x$ vertices of color $b$ (denoted by $b^x_1, b^x_2, \dots, b^x_{n_x}$), and $2n_x$ vertices of color $v$ (denoted by $v^x_1, v^x_2, \dots, v^x_{n_x}$ and $w^x_1, w^x_2, \dots, w^x_{n_x}$). Intuitively, the vertices $v^x_i$ are associated with $x$, and the vertices $w^x_i$ with $\neg x$.
\end{enumerate}

\noindent We now show how to construct the set of edges $E$.

\begin{enumerate}
\item For each variable $x$, we construct a cycle of length $4n_x$ by adding to $E$ the collection of edges $(a^x_i, v^x_i)$, $(v^x_i, b^x_i)$, $(b^x_i, w^x_i)$ and $(w^x_i, a^x_{(i\ \textrm{mod}\ n_x) + 1})$ for $i=1,..,n_x$. 
\item For each clause we add to $E$ three edges, where each edge connects the vertex $c_i$ representing the clause with a vertex representing one literal of $c_i$. More formally, if a literal $x$ ($\neg x$) occurs in the $i$-th clause, we add to $E$ an edge connecting $c_i$ with some vertex  $v^x_j$ ($w^x_j$, respectively). We do this operation in such a way, that each vertex $v^x_j$ and $w^x_j$ representing a literal is incident with at most one clause-vertex $c_i$. Notice that since we have more vertices $v^x_j$ and $w^x_j$ than actual literals, some of the vertices $v^x_j$ and $w^x_j$ will not be connected with any clause-vertex $c_i$.
\end{enumerate}

\subsection{Analysis of the Reduction}

Let $\phi$ be a MAX-3SAT($\beta$) formula on $m$ clauses, and $G=(V,E)$ a vertex-colored graph obtained from $\phi$ by our reduction. Let $G'=(V,E')$ be a subgraph of $G$ which is an optimal solution for the MEC problem on $G$.

\begin{lemma}\label{lem-mec-1}
If the formula $\phi$ is satisfiable, then the transitive closure of $G'$ has at least $12m$ edges.\footnote{It can be proven that in this case the transitive closure of $G'$ has exactly $12m$ edges, but that is not needed in the later part of the reasoning.}
\end{lemma}

\begin{lemma}\label{lem-mec-2}
If any assignment can satisfy at most a $(1-\epsilon)$ fraction of the $m$ clauses of the formula $\phi$, then the transitive closure of $G'$ has at most $12m-\Theta(\epsilon)m$ edges.
\end{lemma}

\begin{theorem}
The Maximum Edges in the Transitive Closure problem is $ APX$-hard, even for graphs with only four colors.
\end{theorem}

\begin{proof}
Let $\phi$ be a MAX-3SAT($\beta$) formula on $m$ clauses, and $G=(V,E)$ a vertex-colored graph obtained from $\phi$ by our reduction. Let $G'=(V,E')$ be a subgraph of $G$ which is an optimal solution for the MEC problem on $G$. From Lemma \ref{lem-mec-1} we know, that if the formula $\phi$ is satisfiable, then the transitive closure of $G'$ has at least $12m$ edges. From Lemma \ref{lem-mec-2} we know, that if any assignment can satisfy at most a $(1-\epsilon)$ fraction of the $m$ clauses of $\phi$, then the transitive closure of $G'$ has at most $12m-\Theta(\epsilon)m$ edges.

As the $MAX-3SAT(\beta)$ problem is $ APX$-hard~\cite{Ausiello1999}, we obtain that MEC is also $ APX$-hard. \qed
\end{proof}


\section{Hardness of MCC}
\label{MCC}
In this section we prove that the MCC problem does not admit polynomial-time approximation within a factor of $|V|^{1/14 - \epsilon}$, for any $\epsilon > 0$, unless $ P = NP$, or within a factor of $|V|^{1/2 - \epsilon}$, unless $ ZPP = NP$. The results hold even if each vertex color appears at most two times in the input graph. We prove our results via a reduction from the Minimum Clique Partition problem.

\vskip5pt
\noindent \textbf{Minimum Clique Partition:} Given a simple, undirected graph $G=(V,E)$, find a partition of $V$ into a minimum number of subsets $V_1, \ldots, V_k$ such that the subgraph of $G$ induced by each set of vertices $V_i$ is a complete graph.
\vskip5pt

The Minimum Clique Partition problem is equivalent to Minimum Graph Coloring~\cite{PM81}, and therefore it cannot be approximated in polynomial time within a factor of $|V|^{1/7 - \epsilon}$ for any $\epsilon > 0$~\cite{BGS98}, unless $ P = NP$, or within a factor of $|V|^{1 - \epsilon}$, unless $ ZPP = NP$~\cite{Feige98}. 


\subsection{Reduction from Minimum Clique Partition}

Let $G = (V,E)$ be an instance of the Minimum Clique Partition problem. We create an instance of the MCC problem, i.e., a vertex colored graph $G' = (V',E')$, as follows. The reduction is illustrated in Figure~\ref{fig:mcc}.

\begin{figure}[t]
\begin{center}
\includegraphics[scale=1]{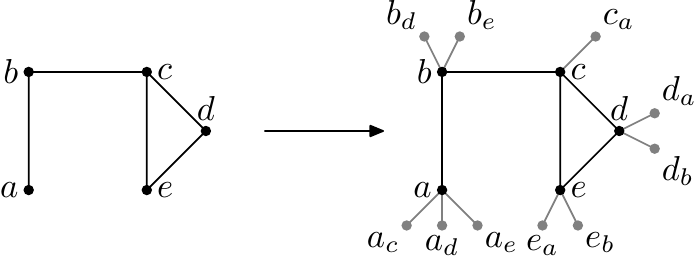}
\end{center}
\caption{Creating an instance of the MCC problem (right) from an instance of the
Minimum Clique Partition (left). Base vertices and edges are drawn in
black, and the additional ones in gray. An optimal solution for both
problems is obtained by removing an edge $(b,c)$.}
\label{fig:mcc}
\end{figure}

\begin{enumerate}
\item The vertex set $V' = V'_b \cup V'_a$ consists of two parts. Firstly, the set $V'_b=V$ is the set of all vertices in $G$, each colored with a distinct color. We term these vertices \emph{base vertices}. 
The set $V'_a$ has two vertices, $u_v$ and $v_u$, for each pair of vertices $u,v \in V$ such that $(u,v) \notin E$. Both vertices $u_v$ and $v_u$ have the same color, which is different from other colors in the graph. We refer to the vertices from $V'_a$ as \emph{additional vertices}.  We emphasize that each color appears \emph{at most} two times in $G'$.

\item The set of edges $E' = E'_b \cup E'_a$ consists of two parts. First, $E'_b = E$ is the set of edges in $G$, which we term \emph{base edges}.  
The set $E'_a$ has two edges, $(u_v,u)$ and $(v_u,v)$, for each pair of vertices $u,v \in V$ such that $(u,v) \notin E$ (i.e., each additional vertex $u_v$ is connected with a base vertex $u$). We refer to the edges from $E'_a$ as \emph{additional edges}. 
\end{enumerate}

\subsection{Analysis of the Reduction}

We first show that the cost of an optimal solution for an instance of the Minimum Clique Partition problem is the same as the cost of an optimal solution of an instance of the MCC problem obtained by the reduction.

\begin{lemma}\label{lem:mcc-red-1}
Let $G=(V,E)$ be an instance of the Minimum Clique Partition problem, and $G'=(V',E')$ the corresponding instance of the MCC problem, obtained by our reduction. If there is a partition of $G$ into $k$ cliques, then the optimal solution for the MCC problem for $G'$ has cost at most $k$.
\end{lemma}

\begin{lemma}\label{lem:mcc-red-2}
Let $G=(V,E)$ be an instance of the Minimum Clique Partition problem, and $G'=(V',E')$ the corresponding instance of the MCC problem, obtained by our reduction. If the optimal solution for the MCC problem for $G'$ has cost $k$, then there exists a partition of $G$ into $k$ cliques.
\end{lemma}

\begin{theorem}
\label{lem:mcc-thm}
The Minimum Colorful Components problem does not admit polynomial time approximation within a factor of $n^{1/14 -\epsilon}$, for any $\epsilon > 0$, unless $ P = NP$, or within a factor of  $n^{1/2 -\epsilon}$, for any $\epsilon > 0$, unless $ ZPP = NP$, where $n$ is the number of vertices in the input graph.
\end{theorem}


\section{Conclusions and future work}

In this paper we study the Colorful Components framework, which arises from applications in biology. We study three problems from this framework: Minimum Singleton Vertices, Maximum Edges in Transitive Closure and Minimum Colorful Components. First, we show a polynomial time exact algorithm for MSV, thus disproving the conjecture of Zheng et al.~\cite{ZSLS11} that the problem is $ NP$-hard. Then, we prove and strengthen another conjecture in~\cite{ZSLS11}, by showing that MEC is $NP$-hard and $APX$-hard. Finally, we show that MCC does not admit polynomial time approximation within a factor of $|V|^{1/14 - \epsilon}$, for any $\epsilon > 0$, unless $ P = NP$, or within a factor of $|V|^{1/2 - \epsilon}$, unless $ ZPP = NP$.

Notice that the $ APX$-hardness result for the MEC problem requires that the input graphs are colored with at least $4$ colors. A natural question is, thus, to settle the complexity of the problem for $3$ colors (as for the case of two colors MEC is easily solvable in polynomial time, using a maximum matching algorithm). Another open question is to design approximation algorithms for the MEC problem or to strengthen the hardness of approximation result.

From the biological perspective it is interesting to analyze how our MSV algorithm behaves on real data. Finally, we mention that an intriguing and challenging task is to find others problems in this framework that admit practical algorithms and are also meaningful for the biological applications.

\bibliographystyle{plain}
\bibliography{bibliography}

\newpage
\appendix
\section{Proofs from Section \ref{sec:poly_algo}}

\begin{proof}[of Lemma \ref{lem:msv_analysis_stop}]
If procedure {\sc Alternating\_Path} returns {\sc no\_path\_found}, then it returns in line $19$, i.e., checking the condition ``$|N'| > 0$'' (line $4$) failed. We show that just before the procedure ends, the following inequality holds:
$$|V'| - |N_{C \setminus\{c\}}(V')| \ge |S_c| \enspace.$$

If the loop in line $4$ has never been entered, we have $V'=S_c$, $N_{C \setminus\{c\}}(V') = N' = \emptyset$, and therefore $|V'| - N_{C \setminus\{c\}}(V') = |S_c|$.

Each vertex $v \in N_{C \setminus\{c\}}(V')$ has been inserted into $N'$ at some step of the procedure, and subsequently processed either in line $5$ or line $9$. As that did not cause the algorithm to return in line $7$ or $11$, we must have:
\begin{itemize}
\item $v$ is not a leaf of a star in $G'$, and
\item the connected component containing $v$ contains a vertex colored with $c$.
\end{itemize}
As each connected component in $G'$ is a singleton, an edge or a star, and the color of $v$ is different from $c$ (from the definition of $N_{C \setminus\{c\}}(V')$), we have the following two possibilities:
\begin{itemize}
\item the connected component of $G'$ containing $v$ is an edge, and the other endpoint of the edge has color $c$, or
\item the connected component of $G'$ containing $v$ is a star containing a vertex of color $c$, and $v$ is the center of the star.
\end{itemize}
We get that any two elements of $N_{C \setminus\{c\}}(V')$ are in different connected components of $G'$, and each vertex $v \in N_{C \setminus\{c\}}(V')$ has some neighbor $n(v)$ in $G'$. Each vertex $n(v)$ has been added to the set $V'$ when the element $v$ has been processed by the procedure (line $13,15$). We get that any two elements $v_1,v_2 \in N_{C \setminus\{c\}}(V')$ are in different connected components of $G'$, any two vertices $n(v_1),n(v_2)$ are different. As the elements from $S_c$ are singletons in $G'$, and therefore cannot be equal $n(v)$, and $S_c \subseteq V'$, we get $|V'| \ge |S_c| + |N_{C \setminus\{c\}}(V')|$. We obtained the desired inequality.

We have shown that for the set of vertices $V' $ we have $|V'| - N_{C \setminus\{c\}}(V') \ge |S_c|$. As $V' \subseteq V_c$, we get $|S_c| \le \max_{V'' \in V_c}(|V''| - N_{C \setminus\{c\}}(V'')) = s_c$. As $s_c$ is a lower bound on $|S_c|$ (see Lemma \ref{lem:MSV_lb}), we get $|S_c| = s_c$.
\qed
\end{proof}


\begin{proof}[of Lemma \ref{lem:msv_analysis_graph}]
$a)$ First let us show that the procedure {\sc Path\_From(v)} always returns a finite path, and such that the first vertex of $p$ in $S_c$. Any vertex which is assigned to the set $N'$ in line $2$ is assigned a predecessor from the set $S_c$ (line $3$). Any vertex $v$ assigned to $N'$ later, i.e., in some $i$-th iteration of the loop (line $16$), is assigned a predecessor pred($v$) $\in V'$, and such that pred($v$) entered $V'$ in the same $i$-th iteration of the loop. Any vertex $v \in V' \setminus S_c$ enters the set $V'$ in some $i$-th iteration of the loop (line $13,15$), and then is assigned a predecessor pred($v$) $\in N'$, such that pred($v$) has been assigned to $N'$ in the previous iteration of the loop (or in line $2$, in case $i=1$). This shows that the procedure {\sc Path\_From(p)} does not loop, and it will eventually (i.e., after at most $|V|$ steps) find a beginning of a path, which is a vertex from the set $S_c$.

We will now show that every odd vertex of the path $p$ is in $V_c$ (except possibly of the last vertex of $p$, if it is the vertex $v'$ appended to the path directly by the procedure in line $7$), and every even vertex is in $V \setminus V_c$. We already know that the first vertex of the path is in $S_c \subseteq V_c$. As all vertices from the set $V'$ have color $c$, and all vertices from the set $N'$ have color different from $c$, a predecessor of a vertex from $V'$ is in $N'$ (line $14$) and a predecessor of a vertex from $N'$ is in $V'$ (line $17$), the claim follows. Notice that if the procedure reaches line $7$, then there are no color requirements for the last vertex $v'$ appended at the end of the path directly, and not using the procedure {\sc path\_from} (line $7$).

We will now show that every even edge of the path $p$ is in $E'$, and every odd edge of $p$ is in $E \setminus E'$, which will prove that $p$ is an alternating path. Let us consider even edges first. An even edge is an edge between some odd vertex $v$ and a preceding vertex $w$. There can be two cases, and for both of them we obtain that the edge is in $E'$:
\begin{itemize}
\item $v$ is a vertex appended to the path directly by the procedure in line $7$. Then the edge connecting $v$ with the preceding vertex $w$ is in $E'$ (see line $7$).
\item $v$ has been appended to the path by the procedure {\sc path\_from}). Then $w=$pred($v$) and, from the paragraph above, $v \in V_c$. A vertex from $V_c$ is connected with its predecessor via an edge in $E'$ (see line $14$).
\end{itemize} 

Now let us consider odd edges. As the path $p$ starts in a singleton vertex of $G'$, the first (odd) edge of the path is not in $E'$. Let $(v',v)$ be any other odd edge of $p$. We have $v'=\textrm{pred}(v)$, $v' \in V_c$. Let $w=\textrm{pred}(v')$. As $w$ has been processed by the procedure in an earlier loop than $v$ (see the first paragraph of the proof), and processing $w$ did not cause the procedure to return in line $7$ or $11$, one of the following holds (as each connected component of $G'$ is a singleton, an edge or a star):
\begin{itemize}
\item the connected component of $G'$ containing $w$ is an edge, and the other endpoint of the edge has color $c$, or
\item the connected component of $G'$ containing $w$ is a star containing a vertex of color $c$, and $v$ is the center of the star.
\end{itemize}
As $(w,v')$, as an even edge of the path, belongs to $E'$, that gives us that $v'$ has degree one is $G'$ (either as an endpoint of an edge, or a leaf of a star), and so $(v',v) \notin E'$.
\vskip10pt

$b)$ From $a)$ we know that the path $p$ applied to $G'$ to construct $G''$ is an alternating path, i.e., after applying it the degree of each vertex other than the endpoints of the path does not change. The path starts with a vertex $v \in S_c$ (see $a)$), which is a singleton in $G'$,
and therefore the first edge of $p$ is not in $E'$. The first edge of $p$ is added to the graph and $v$ stops being a singleton vertex. That decreases the number of singleton vertices of color $c$ by one.

We now have to consider the last edge of the path. Again, if the edge is not in $E'$ then the endpoint of the path gets one additional edge incident with it, and so it cannot become a singleton. The only possibility when the last edge of $p$ is in $E'$ is when the path has an even number of edges (as it is an alternating path starting with an edge in $E \setminus E'$), i.e., it ends with an odd vertex. The procedure {\sc path\_from} is always invoked for a vertex $v \in N'$ (line $6$,$10$) and the path returned by it has an odd number of edges (see the $a)$, where we show that such path alternates between vertices from $V'$ and $N'$). The only possibility that the path has an even number of edges is when it is generated in line $7$, when an additional edge $(v,v')$ is appended at the end of the path. Then the edge $(v,v')$ is removed from $G''$ and the degree of $v'$ drops by one. However, from line $5$ we get that then $v$ is a leaf of a star, and as $(v,v') \in E'$ we have that $v'$ is a center of a star. The degree of $v'$ in $G'$ is at least $2$, so $G'$ does not become a singleton after applying the path $p$ to $G'$.
\vskip10pt

$c)$ As every connected component of $G'$ is colorful, it is enough to consider components of $G''$ which contain some newly added edge. Let $(u,v) \in E$ be an edge added to $G''$, i.e., an edge from $p$ which is in $E \setminus E'$. From the discussion in $a)$ we know that $(u,v)$ is then an odd edge of the path, and it connects an even vertex $v \in V \setminus V_c$ with its predecessor $u=$pred($v$)$\in V_c$.

We will now show that the degree of $u$ in $G'$ is at most one. If $u$ is the start of the path, it has degree $0$ in $G'$. Otherwise, considering the predecessor of $u$ and using the same arguments as in $a)$ we show that $u$ is either an endpoint of a path or a leaf of a star in $G'$. In this case the degree of $u$ in $G'$ is one. As the edge connecting $u$ with its predecessor in $p$ is in $E'$, it will be removed from $G'$. The vertex $u$ is a leaf in the connected component of $G''$. 

If vertex $v$ is not the last even vertex of the path, then, from the construction of $p$ and the discussion in $a)$, the successor of $p$ is some vertex $w$ of color $c$, and the edge $(v,w) \in E'$. Then the connected component of $v$ (which in $G'$ was either an edge or a star centered at $v$, again from the discussion in $a)$) obtains a new vertex $v$ of color $c$, but on the other hand loses some other vertex $w$ of color $c$. The component remains colorful.

If vertex $v$ is the last even vertex of the path, and the procedure returned in step $7$ after processing $v$, $v$ has been detached from its component in $G'$ (which was a star), and the new component is an edge connecting $u$ and $v$, and it is colorful.

Finally, if vertex $v$ is the last even vertex of the path, and the procedure returned in step $11$ after processing $v$, the connected component of $v$ in $G'$ did not contain vertex of color $c$, so a new vertex of color $c$ can be attached to it and the component remains colorful.  
\vskip10pt

$d)$ We show it similarly as $c)$. As every connected component of $G'$ is either a singleton, an edge, or a star, and removing the edges does not change this property, it is enough to consider components of $G''$ which contain some newly added edge $(u,v) \in E \setminus E'$. As in case $c)$, the degree of $u$ in $G'$ is at most one, and $u$ becomes a leaf in the connected component of $G''$.

Considering the same three cases as in $c)$ we have that either:
\begin{itemize}
\item $v$ is not the last even vertex of the path: then the connected component of $v$ (which in $G'$ was either an edge or a star centered at $v$) gets one leaf attached at $v$, at the same time losing another leaf attached at $v$, or
\item $v$ is the last even vertex of the path, and the procedure returned in step $7$ after processing $v$: new component is an edge connecting $u$ and $v$, or
\item $v$ is the last even vertex of the path, and the procedure returned in step $11$ after processing $v$: in this case, as $v$ has not been a leaf of a star, it could either be a singleton, an endpoint of an edge, or a center of a star; in any of these cases attaching a leaf to $v$ makes the connected component an edge or a star.
\end{itemize}
The connected component containing the edge $(u,v)$ is either an edge or a star.
\qed
\end{proof}

\section{Proofs from Section~\ref{MEC}}

\begin{proof}[of Lemma~\ref{lem-mec-2}]
To prove the lemma it is enough to show that if the transitive closure of $G'$ has more than $12m-\epsilon m$ edges, we can extract from $G'$ an assignment $f$ for $\phi$ which satisfies at least a $1-O(\epsilon)$ fraction of clauses. For the rest of the proof we assume that the transitive closure of $G'$ has more than $12m-\epsilon m$ edges.

First, observe that each connected component of $G'$ has size at most $4$, since there are only $4$ colors of the vertices in the graph. Also, notice that there are in total $m$ vertices of color $c$, $6m$ vertices of color $v$ (as the total number of literals in the formula $\phi$ equals $3m$), $3m$ vertices of color $a$, and $3m$ vertices of color $b$.

We  now show that $G'$ has at least $(1-\epsilon)m$ connected components consisting of $4$ vertices. Let $\alpha_1, \alpha_2, \alpha_3$ and $\alpha_4$ denote the number of connected components of $G'$ of size $1,2,3$ and $4$, respectively. The number of edges in the transitive closure of $G'$ equals $\textrm{OPT}=6\alpha_4+3\alpha_3+\alpha_2$. As $G'$ has $7m$ vertices of color other than $v$, exactly $3$ such vertices are in each component of size $4$, at least two such vertices are in each component of size $3$, and at least one such vertex is in each component of size $2$, we get: $\alpha_3 \le (7m-3\alpha_4)/2$ and $\alpha_2 \le (7m - 3\alpha_4 - 2\alpha_3)$. We get
$$\textrm{OPT}=6\alpha_4+3\alpha_3+\alpha_2 
\le 7m+3\alpha_4+\alpha_3
\le 10.5m + 1.5\alpha_4\enspace.$$
As we assumed $\textrm{OPT} > 12m-\epsilon m$, we get that $\alpha_4 \ge (1-\epsilon)m$.

Let us now consider a subgraph $G'_x$ of $G'$ corresponding to the variable $x$. $G'_x$ consists of vertices $a^x_i, b^x_i, v^x_i$ and $w^x_i$ for $1=1,\ldots ,n_x$, and additionally of the clause-vertices $c_j$ which are incident in $G'$ with any of the vertices $v^x_i$ and $w^x_i$. As each clause-vertex has degree at most $1$ in $G'$ (as all the neighbors of a clause-vertex have the same color $v$), it belongs to at most one subgraph $G'_x$. Notice that each edge of $G'$ is contained in some subgraph $G'_x$, and therefore the edges of the transitive closure of $G'$ are the union of the edges of the transitive closures of $G'_x$.

\begin{figure}[t]
\begin{center}
\includegraphics[scale=1]{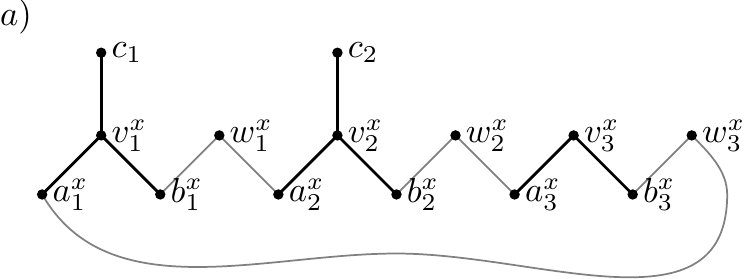}
\includegraphics[scale=1]{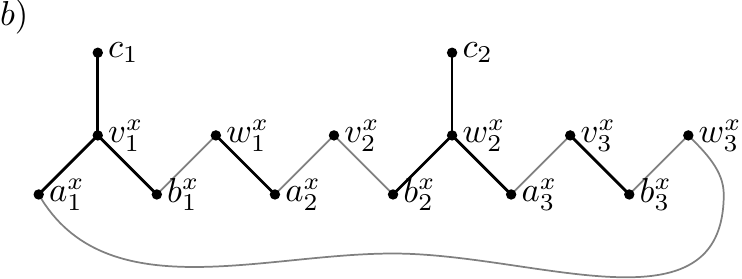}
\end{center}
\caption{$a)$ Graph $G'_x$ maximizing the possible number of edges ($3 \alpha^x + 3 n_x$) in the transitive closure. $b)$ An inconsistent graph $G'_x$ cannot achieve $3 \alpha^x + 3 n_x$ edges in the transitive closure.}
\label{fig:mec-reduction-2}
\end{figure}

We say that $G'_x$ is \emph{inconsistent} if there are two vertices $v^x_i$ and $w^x_j$ in $G'_x$ (where possibly $i=j$), such that both are incident with some clause-vertices $c_{i'}$ and $c_{j'}$ in $G'$. We  now show that $G'$ has at most $\epsilon m$ inconsistent subgraphs $G'_x$. Let $\alpha^x \le n_x$ be the number of clause-vertices $c_i$ which belong to $G'_x$. Apart from these $\alpha^x$ vertices of color $c$, $G'_x$ has also $n_x$ vertices of colors $a$ and $b$, and $2n_x$ vertices of color $v$. It is straightforward to verify that the transitive closure of $G'_x$ has at most $3\alpha^x + 3n_x$ edges (see Figure \ref{fig:mec-reduction-2}$a$ for an example, where this bound is obtained). Moreover, if $G'_x$ is inconsistent, its transitive closure has at most $3\alpha^x + 3n_x-1$ edges (see Figure \ref{fig:mec-reduction-2}$b$). 
If $G'$ has more than $\epsilon m$ inconsistent subgraphs $G'_x$, we have $\textrm{OPT} \le \sum_x 3 \alpha^x + \sum_x 3n_x - \epsilon m = 12m - \epsilon m$, which contradicts our assumption.

We  now extract from $G'$ an assignment $f$ for $\phi$ which satisfies at least a $(1-\epsilon(\beta+1))$ fraction of clauses. We proceed as follows. From our previous reasoning we know that $G'$ has at least $(1-\epsilon)m$ connected components consisting of $4$ vertices, and that at most $\epsilon m$ subgraphs $G'_x$ are inconsistent. We fix the assignment $f$ as follows. For each variable $x$, if $G'_x$ is inconsistent or if $G'_x$ does not contain any vertices $c_i$, we set $f(x)$ arbitrarily. Otherwise, either \emph{all} vertices $c_i$ from the component $G'_x$ are incident with vertices $v^x_j$ (corresponding to the literal $x$), or they are all incident with vertices $w^x_j$ (corresponding to the literal $\neg x$). If the first case holds, we set $f(x)$ to TRUE. Otherwise, we set $f(x)$ to FALSE. 

We  now show a lower bound on the number of clauses satisfied by $f$. At least $(1-\epsilon)m$ clause-vertices $c_i$ are incident with some variable-vertex (as there are at least $(1-\epsilon)m$ connected components of size $4$). Each variable occurs at most $\beta=O(1)$ times in $\phi$, and at most $\epsilon m$ subgraphs $G'_x$ are inconsistent, and therefore at most $\epsilon \beta m$ clause-vertices are incident with variables from an inconsistent subgraph. Therefore at least $m(1-\epsilon(\beta+1))$ clauses are satisfied by the assignment $f$. \qed
\end{proof}

\begin{proof}[of Lemma~\ref{lem-mec-1}]
We construct a graph $G''=(V,E'')$ which is a subgraph of $G$ in the following way (see Figure~\ref{fig:mec-reduction}). Fix a satisfying assignment $f$ for $\phi$. For each clause, represented by a vertex $c_i$, we choose arbitrarily a literal $x$ ($\neg x$) which is satisfied by the assignment $f$. Let $v^x_j$ ($w^x_j$, respectively) be the vertex corresponding to the chosen literal which is incident with $c_i$ in $G$. We add the edge $(c_i,v^x_j)$ ($(c_i,w^x_j)$, respectively) to $G''$. Additionally, each vertex $v^x_j$ and $w^x_j$ associated with a literal satisfied by $f$ is connected in $G''$ with the neighboring vertices of color $a$ and $b$. 

It is straightforward to check that $G''$ is a feasible solution for the MEC problem (i.e., each connected component of $G''$ is colorful), and that $G''$ has $m$ connected components containing $4$ vertices, $2m$ connected components containing $3$ vertices, and $3m$ singletons. The transitive closure of $G''$ has $6\cdot m+3\cdot 2m = 12m$ edges. As $G'$ is an optimal solution for the MEC problem in $G$, the transitive closure of $G'$ has at least as many edges as the transitive closure of $G''$. \qed
\end{proof}

\section{Proofs from Section~\ref{MCC}}

\begin{proof}[of Lemma~\ref{lem:mcc-red-1}]
Let $G$ be a graph which can be partitioned into $k$ cliques. We have to show that there is a collection of edges $E'' \subseteq E'$ in $G'$, such that after removing $E''$ from $G'$ we obtain a graph consisting of at most $k$ colorful components. The set of edges $E''$ is exactly the set of base edges that have been removed from $G$ to obtain the collection of $k$ cliques. 

As we do not remove any additional edges of $G'$ (i.e., the edges from the set $V'_a$), the resulting graph consists of $k$ connected components. The only pairs of vertices sharing the same color are pairs $u_v, v_u$ such that $u,v \in V$ and $(u,v) \notin E$. Then $u$ and $v$ must be in different connected components of the clique partition, and so $u$ and $v$ (and therefore also $u_v$ and $v_u$) are in different connected components of the constructed graph. Each connected component of the constructed graph is colorful. \qed
\end{proof}

\begin{proof}[of Lemma~\ref{lem:mcc-red-2}]
Let $G'$ be a graph which can be transformed, by removing a collection of edges $E'' \subseteq E'$, into a graph consisting of $k$ connected colorful components. We  show that we can modify $E''$, without increasing the number of connected components in the resulting graph and while ensuring that each connected component stays colorful, so that $E''$ does not contain any edge from the set of additional edges $E'_a$. Then, by removing $E''$ from $G$, we obtain a valid partition of $G$ into at most $k$ cliques: For any pair of vertices $u,v \in V$ such that $(u,v) \notin E$, there are two vertices $u_v$ and $v_u$ in $G'$, sharing the same color and connected with $u$ and $v$, respectively, via additional edges. As no additional edges are contained in $E''$ and each connected component of $(V',E' \setminus E'')$ is colorful, $u$ and $v$ must be in different connected components of $(V',E' \setminus E'')$. In the partition of $G$ the vertices $u$ and $v$ are then also in different connected components, and so we obtain a partitioning of $G$ into at most $k$ cliques.

We now show how to modify $E''$. For each additional edge $e=(u,u_v) \in E'_a$ which is in $E''$ we perform the following operation. First, we remove $e$ from $E''$. That decreases the number of connected components by one, but might result in an infeasible solution. However, the only pair of vertices of the same color which are in the same connected component of $(V',E' \setminus E'')$ can now be $u_v$ and $v_u$. Denote by $C$ the connected component of $(V',E' \setminus E'')$ containing $u_v$ and $v_u$, and therefore also $u$ and $v$. We now find a minimum cut separating $u$ from $v$ in $C$, and add the edges of the cut to $E''$. That results in splitting $C$ into exactly two connected components, and each of the two components is colorful. 

We perform the above operation for each additional edge $e=(u,u_v) \in E''$, and at the end we obtain a set $E''$ satisfying the needed conditions. By removing $E''$ from $G$ we get a partition of $G$ into at most $k$ cliques. \qed
\end{proof}

\begin{proof}[of Theorem~\ref{lem:mcc-thm}]
Let $G=(V,E)$ be an instance of the Minimum Clique Partition problem, and $G'=(V',E')$ the corresponding instance of the MCC problem, obtained by our reduction. From Lemmas \ref{lem:mcc-red-1} and \ref{lem:mcc-red-2} we obtain, that the cost of an optimal solution for the MCC problem for $G'$ is the same as the cost of an optimal solution for the Minimum Clique Partition problem for $G$. We know that the Minimum Clique Partition problem is hard to approximate within a factor of $|V|^{1/7 - \epsilon}$, unless $ P = NP$, or  within a factor of  $|V|^{1/2 - \epsilon}$, unless $ ZPP = NP$,  were $|V|$ is the number of vertices in the graph $G$. Since $G'$ has $|V'| \le |V|^2$ vertices, our theorem follows. \qed
\end{proof}

\end{document}